\documentclass[11pt]{amsart}

\usepackage{amsmath}
\usepackage{amssymb}
\usepackage{amsfonts}
\usepackage{latexsym}
\usepackage{color}

\newtheorem{theo}{Theorem}[section]
\newtheorem{defn}[theo]{Definition}
\newtheorem{lem}[theo]{Lemma}
\newtheorem{prop}[theo]{Proposition}
\newtheorem{cor}[theo]{Corollary}
\newtheorem{lemma}[theo]{Lemma}

\newtheorem{example}[theo]{Example}
\newtheorem{remark}[theo]{Remark}

\newcommand{\ad}{^{\mbox{\scriptsize $\dag$}}}

\newcommand{\up}{\raisebox{0.7mm}{$\upharpoonright \,$}}%

\newcommand{\be}{\begin{equation}}
\newcommand{\en}{\end{equation}}
\newcommand{\bea}{\begin{eqnarray}}
\newcommand{\ena}{\end{eqnarray}}
\newcommand{\beano}{\begin{eqnarray*}}
\newcommand{\enano}{\end{eqnarray*}}
\newcommand{\bee}{\begin{enumerate}}
\newcommand{\ene}{\end{enumerate}}
\newcommand{\bedefi}{\begin{defn}$\!\!${\bf }$\;$\rm }
\newcommand{\findefi}{\end{defn}}

\def\A{{\mathfrak A}}

\def\D{{\mathcal D}}

\def\H{{\mathcal H}}

\def\L{{\mathcal L}}

\def\R{{\mathcal R}}

\def\J{\relax\ifmmode {\mathcal J}\else${\mathcal J}$\fi}
\def\x{\relax\ifmmode {\mbox{*}}\else*\fi}

\def\BB{{\mathfrak B}}

\newcommand{\mc}{\mathcal}
\newcommand{\mb}{\mathbb}

\newcommand{\mult}{{\scriptstyle \Box}}

\newcommand{\ha}{^{\rm\textstyle *}}

\newcommand{\pa}{partial \mbox{*-algebra}}

\newcommand{\LDH}{{\L}\ad(\D,\H)}

\def\dag{\dagger}

\newcommand{\vp}{\varphi}

\newcommand{\ip}[2]{\left\langle{#1}\right|\left.{#2}\right\rangle}
\newcommand{\Dr}{\D_\pi}

\newcommand{\ze}{_{\scriptscriptstyle 0}}

\newcommand{\betheo}{\begin{theo}}
\newcommand{\entheo}{\end{theo}}

\newcommand{\becor}{\begin{cor}}
\newcommand{\encor}{\end{cor}}

\newcommand{\belem}{\begin{lem} }
\newcommand{\enlem}{\end{lem}}

\newcommand{\beprop}{\begin{prop} }
\newcommand{\enprop}{\end{prop}}
\newcommand{\beex}{\begin{example}$\!\!\!${\bf} \rm }
\newcommand{\enex}{ \end{example}}

\newcommand{\berem}{\begin{remark}$\!\!\!${\bf} \rm }
\newcommand{\enrem}{ \end{remark}}

\begin{document}
\title[Representable functionals]{Representable linear functionals\\ on partial *-algebras}

\date{\today}
\author{F. Bagarello}
\address{Dipartimento di Metodi e Modelli Matematici,
Fac. Ingegneria, Universit\`a di Palermo, I-90128  Palermo, Italy}
\email{bagarell@unipa.it}

\author{A. Inoue}
\address{
Department of Applied Mathematics, Fukuoka University, Fukuoka
814-0180, Japan}
\email{a-inoue@fukuoka-u.ac.jp}
\author{C. Trapani} \address{Dipartimento di Matematica ed
Applicazioni, Universit\`a di Palermo, I-90123 Palermo
Italy}
\email{
trapani@unipa.it}

\begin{abstract}
A GNS - like *-representation of a \pa\ $\A$ defined by certain representable linear functionals on $\A$ is constructed.
The study of the interplay with the GNS construction associated with invariant positive sesquilinear forms (ips) leads to the notions of pre-core and of singular form. It is shown that a positive sesquilinear form with pre-core always decomposes into the sum of an ips form and a singular one.
\end{abstract}

\maketitle

\section{Introduction and Preliminaries}

The Gelfand - Naimark - Segal (GNS) representation plays, as it is well known, a key role in the study of the structure of topological *-algebras and it is important in many physical applications. Since from the very beginning when partial *-algebras and quasi *-algebras were studied, it was clear that an extension of the GNS construction was needed. The most natural solution consisted in considering, as starting point, certain positive sesquilinear forms, called ips (invariant positive sesquilinear) forms and biweights, since they allowed to by-pass the lack of a noneverywhere defined multiplication. In this framework, and following the different axioms that were introduced from time to time, several types of GNS-like representations were constructed taking their values in partial *-algebras of unbounded operators (partial O*-algebras) or in families of operators in  partial inner product ({\sc PIP}) spaces . We refer to the monographs \cite{ait_book, at_pipbook} for an overview and for complete references on this subject.

However, in many situations it is more interesting, and sometimes more natural,  to consider the possibility that a linear functional $\omega$ on a \pa\ $\A$, positive in some sense, could be taken as basic ingredient of the construction. This is the case, for instance, of applications to quantum theories, where linear functionals describe equilibrium states at given temperature (Gibbs states for finite systems, KMS states for infinite systems).
We will show (Section \ref{sect_GNS}) that a GNS-like construction is, indeed, possible starting from certain linear functionals on $\A$, called {\em representable}.

In Section \ref{sect_biwips} we consider the interplay between representable linear functionals and ips forms (or biweights). For this aim, we introduce the notion of {\em pre-core} for a positive sesquilinear form $\vp$ on $\A \times \A$ ($\A$ a semi-associative \pa\ with unit) and we show that if a pre-core exists, there is a representable linear functional $\omega_\vp$ which is associated to it in natural way. The construction of Section \ref{sect_GNS} can then be performed, even though $\vp$ is {\em not}, in general, an ips form on $\A$ and so a GNS representation starting directly from $\vp$ cannot be defined. This fact leads (Section \ref{sect_decom}) to the definition of {\em singular form} where the ``singularity'' consists, essentially, in the failure of a density condition, needed in the definition of ips form.
The outcome is a decomposition theorem, which can be viewed as the main result of this paper: every positive sesquilinear form $\vp$ with pre-core can  be decomposed into the sum of an ips form and of a {singular} form.
Finally, in Section \ref{sect_quasireg} we analyze the behavior of positive vector sesquilinear forms $\vp_\xi^\pi$ defined via a given *-representation $\pi$ of $\A$ and a vector $\xi \in \Dr$, the domain of $\pi$. In contrast with the case of *-algebras, such a form does not necessarily allow a GNS contruction. When this happens for every vector $\xi \in \Dr$ (the domain of $\pi$), we speak of a {\em quasi regular} *-representation. Conditions for the quasi regularity of a *-representation $\pi$ are given.

\medskip Throughout this paper we follow the definitions and notations given in \cite{ait_book}. We
simply recall that a \pa\ $\A$ is a complex vector space with conjugate linear involution  $\ha $
and a distributive partial multiplication $\cdot$, defined on a subset $\Gamma \subset \A \times \A$, satisfying
the property that $(x,y)\in \Gamma$ if, and only if, $(y\ha ,x\ha )\in   \Gamma$ and $(x\cdot y)\ha = y\ha \cdot x\ha $.
>From now on we will write simply $xy$ instead of $x\cdot y$ whenever $(x,y)\in \Gamma$. For every $y \in \A$, the
set of left  (resp. right) multipliers of $y$ is denoted by $L(y)$ (resp. $R(y)$),  i.e., $L(y)=\{x\in
\A:\, (x,y)\in \Gamma\}$. We denote by $L\A$ (resp. $R\A$)  the space of universal left (resp. right) multipliers
of $\A$.

In general a \pa\ is not associative, but in several situations a weaker form of associativity holds. More precisely, we say
that $\A$ is \emph{ semi-associative} if $b \in R(a)$ implies $by\in R(a)$, for every $y \in R\A$ and
 $$
(ab)y=a(by).
$$
We notice that, if $\A$ is semi-associative and $a\in L(b)$, then $x^*a \in L(by)$, for every $x,y \in R\A$ and
\begin{equation} \label{assoc} x^*(ab)y = (x^*(ab))y = (x^*a)(by), \quad \forall x,y \in R\A. \end{equation}
Furthermore, in this case, $R\A$ is an algebra.\\

\medskip Let $\H$ be a complex Hilbert space and $\D$ a dense subspace of $\H$.
 We denote by $ \L\ad(\D,\H) $
the set of all (closable) linear operators $X$ such that $ {\D}(X) = {\D},\; {\D}(X\x) \supseteq {\D}.$ The set $
\L\ad(\D,\H ) $ is a  \pa\
 with respect to the following operations: the usual sum $X_1 + X_2 $,
the scalar multiplication $\lambda X$, the involution $ X \mapsto X\ad = X\x \up {\D}$ and the \emph{ (weak)}
partial multiplication $X_1 \mult X_2 = {X_1}\ad\x X_2$, defined whenever $X_2$ is a weak right multiplier of
$X_1$ (we shall write $X_2 \in R^{\rm w}(X_1)$ or $X_1 \in L^{\rm w}(X_2)$), that is, iff $ X_2 {\D} \subset
{\D}({X_1}\ad\x)$ and  $ X_1\x {\D} \subset {\D}(X_2\x).$  $\LDH$ is neither associative nor semi-associative.

\medskip
A \emph{ *-representation} of a  \pa\ $\A$ in the
Hilbert space $\H$ is a linear map $\pi : \A \rightarrow\L\ad(\D,\H)$     such that: \begin{itemize}
\item[(i)] $\pi(a\x) = \pi(a)\ad$ for every $a \in \A$;
\item[(ii)] $a \in L(b)$
in $\A$ implies $\pi(a) \in L^{\rm w}(\pi(b))$ and $\pi(a) \mult\pi(b) = \pi(ab).$
\end{itemize}

The {\em closure} $\widetilde{\pi}$ of a *-representation $\pi$ is defined by $$ \widetilde{\pi}(a):= \overline{\pi(a)}\upharpoonright \widetilde{\Dr},$$
where $\widetilde{\Dr}$ is the completion of $\Dr$ under the {\em graph topology} defined on $\Dr$ by the family of seminorms
$$ \xi \to \|\pi(a)\xi\|, \; a \in \A.$$

The *-representation $\pi$ is called {\em closed} if $\Dr = \widetilde{\Dr}$ and {\em  fully closed} if $$
\Dr =\bigcap_{a\in \A}D(\overline{\pi(a)}).
$$

\section{GNS-like construction}\label{sect_GNS}

In this Section we will examine the possibility of a GNS-like construction in \pa s, starting from certain linear functionals. Some of these results generalize those given in \cite{ct_rep} for quasi *-algebras.

>From now on, we will assume that $\A$ is a {\em semi-associative} partial *algebra.

\betheo \label{GNS}
Let $\omega$ be a linear functional on $\A$, $\BB$ a subspace of $R\A$,  satisfying the following requirements:

(R1) $\omega(x^*x)\geq 0$ for all $x\in\BB$;

(R2) $\omega(y^*(a^*x))=\overline{\omega(x^*(ay))}$,
$\forall\,x,y\in\BB$, $a\in\A$;

(R3) $\forall a\in\A$ there exists $\gamma_a>0$ such that
$|\omega(a^*x)|\leq \gamma_a\,\omega(x^*x)^{1/2}$, for all $x \in \BB$.

Then there exists a triple $(\pi^\BB_{\omega}, \lambda^\BB_{\omega}, \H^\BB_{\omega})$ such that

\begin{itemize}
  \item[(a)] $\pi^\BB_{\omega}$ is  a *-representation of $\A$ in $\H_\omega$;
  \item[(b)] $\lambda^\BB_{\omega}$ is a linear map of $\A$ into
  $\H^\BB_{\omega}$ with $\lambda^\BB_{\omega}(\BB)=\D_{\pi^\BB_\omega}$ and
  $\pi^\BB_{\omega}(a)\lambda^\BB_{\omega}(x)=\lambda^\BB_{\omega}(ax)$, for every $a \in\A,\, x \in \BB$.
  \item[(c)] $\omega(y^*(ax))=\ip{\pi^\BB_{\omega}(a)\lambda^\BB_\omega(x)}{\lambda^\BB_\omega(y)}$,
  for every $a \in \A$, $x,y \in \BB$.
\end{itemize}
In particular, if $\A$ has a unit $e$ and $e \in \BB$, we have:
\begin{itemize}
  \item[(a$_1$)] $\pi^\BB_{\omega}$ is a cyclic *-representation of $\A$ with cyclic vector $\xi_\omega$;
  \item[(b$_1$)] $\lambda^\BB_{\omega}$ is a linear map of $\A$ into
  $\H^\BB_{\omega}$ with $\lambda^\BB_{\omega}(\BB)=\D_{\pi^\BB_\omega}$, $\xi_\omega=\lambda^\BB_{\omega}(e)$ and
  $\pi^\BB_{\omega}(a)\lambda^\BB_{\omega}(x)=\lambda^\BB_{\omega}(ax)$, for every $a \in\A,\, x \in \BB$.
  \item[(c$_1$)] $\omega(a)=\ip{\pi^\BB_{\omega}(a)\xi_\omega}{\xi_\omega}$,
  for every $a \in \A$.
\end{itemize}
\entheo
\begin{proof}
We define $N_\omega= \{ x \in \BB:  \, \omega(x^*x)=0\}$. Then (R3) implies that
    $$N_\omega =\{x \in \BB;\, \omega(y^*x)=0, \;\forall y \in \BB\},$$ so that $N_\omega$ is a subspace of $\BB$. The quotient
    $\BB/N_\omega\equiv\{\lambda_\omega^0(x):=x+N_\omega; x \in \BB\}$ is a pre-Hilbert space with inner product
    $$ \ip{\lambda_\omega^0(x)}{\lambda_\omega^0(y)}= \omega(y^*x),
    \quad x,y \in \BB.$$
     Let $\H_\omega$ be the completion of
    $\lambda_\omega^0(\BB)$.\\
    If $a\in \A$, we put $a^\omega(\lambda_\omega^0(x))=
    \omega(a^*x)$. Then, by (R3), it follows that $a^\omega$ is a
    well defined linear functional on $\lambda_\omega^0(\BB)$ and we
    have
    $$ |a^\omega(\lambda_\omega^0(x))|= |\omega(a^*x)| \leq \gamma_a
    \omega(x^*x)^{1/2}= \gamma_a \|\lambda_\omega^0(x)\|, \quad
    \forall x \in \BB .$$
    Thus,  $a^\omega$ extends to a continuous linear functional on $\H^\BB_{\omega}$ and so,
  by Riesz's lemma, there exists a unique $\xi_a \in
    \H^\BB_{\omega}$ such that
    \begin{equation} \label{eqn_defn_lambda} a^\omega(\lambda_\omega^0(x))=
    \ip{\lambda_\omega^0(x)}{\xi_a}, \quad \forall x \in\BB.\end{equation}

Now we put
\begin{equation} \label{eqn_defn_lambda2} \lambda^\BB_{\omega} (a)=\xi_a, \quad a \in \A.\end{equation}
Then $\lambda^\BB_{\omega}$ is a linear map from $\A$ into $\H_\omega$, which extends $\lambda_\omega^0$.
For $a \in \A$, we define
$$ \pi^\BB_\omega(a) \lambda^\BB_{\omega} (x):= \lambda^\BB_{\omega} (ax), \quad x \in \BB.$$
Since,
\begin{eqnarray*}\ip{\lambda^\BB_{\omega}(y)}{\pi^\BB_\omega
(a)\lambda^\BB_{\omega}(x)}&=&\ip{\lambda^\BB_{\omega}(y)}{\lambda^\BB_{\omega}(ax)}
=(ax)^\omega (\lambda^\BB_{\omega}(y))\\ &=&
\omega((x^*a^*)y)=\overline{\omega((y^*a)x)},\quad \forall y \in \BB,
\end{eqnarray*}
it follows from (R3) that $\pi^\BB_\omega(a)$ is well-defined and
maps $\lambda^\BB_{\omega}(\BB)$ into $\H_\omega$. In similar way one
can show the equality
$$\ip{\pi^\BB_\omega
(a^*)\lambda^\BB_{\omega}(y)}{\lambda^\BB_{\omega}(x)}=\overline{\omega((y^*a)x)},
\quad  \forall x, y \in \BB.$$ This implies that $\pi^\BB(a) \in
{\mathcal L}\ad(\lambda^\BB_{\omega}(\BB),\H_\omega)$ and
$\pi^\BB_\omega(a)\ad = \pi^\BB_\omega(a^*)$. \\

Using the semi-associativity of $\A$, in particular  \eqref{assoc}, we also get, for $a,b \in \A$ with $a\in L(b)$,
the equality
$$\ip{ \pi^\BB_\omega(ab)\lambda^\BB_{\omega}(x)}{\lambda^\BB_{\omega}(y)}=
\ip{\pi^\BB_\omega(b)\lambda^\BB_{\omega}(x)}{\pi^\BB_\omega(a^*)\lambda^\BB_{\omega}(y)},
\quad \forall x,y \in \BB.$$ This implies that $\pi^\BB_\omega(a)\mult
\pi^\BB_\omega(b)$ is well-defined and
$$\pi^\BB_\omega(ab)=\pi^\BB_\omega(a)\mult \pi^\BB_\omega(b), \quad \forall a, b
\in \A, a\in L(b).$$ Thus, $\pi^\BB_\omega$ is a *-representation. It is clear that, if $\BB$ contains the unit $e$ of $\A$, then $\pi^\BB_\omega(\BB)\xi_\omega$,
$\xi_\omega:=\lambda^\BB_{\omega}(e)$, is dense in $\H_\vp$. \end{proof}

\berem In general, $\lambda^\BB(\A) \varsubsetneq \D_{\widetilde{\pi}^\BB_\omega}$, the domain of the closure $\widetilde{\pi}^\BB_\omega$ of $\pi^\BB_\omega$, but it is contained in
$ \D_{{(\pi^\BB_{\omega}\upharpoonright \BB)}^*}:=\displaystyle \bigcap_{x\in \BB}D({\pi^\BB_{\omega}(x)}^*)$ and ${\pi^\BB_{\omega}}^*(x)\lambda^\BB_\omega(a)=\lambda^\BB_\omega(xa)$ for every $a\in \A,\,x \in \BB$.
\enrem

The representation $\pi^\BB_\omega$ depends on the choice of the subspace $\BB\subset \R\A$. Thus, it makes sense to compare representations defined by different subspaces of $R\A$.

\beprop Let $\omega$ be a linear functional on $\A$ and $\BB_1,\, \BB_2$ two subspaces of $R\A$ satisfying (R1), (R2), (R3). Suppose that $\BB_1 \subset \BB_2$. Then the *-representation $\pi^{\BB_1}_\omega$ can be regarded as a *-subrepresentation of $\pi^{\BB_2}_\omega$, in the sense that there exists an isometry $U$ of $\H^{\BB_1}_\omega$ into $\H^{\BB_2}_\omega$ such that $U\lambda^{\BB_1}_\omega (x)= \lambda^{\BB_2}_\omega (x)$, for every $x \in \BB_1$, and
$U^*\pi^{\BB_2}_\omega(a)U \lambda^{\BB_1}_\omega (x)= \pi^{\BB_1}_\omega(a)\lambda^{\BB_1}_\omega (x)$, for every $a \in \A$ and $x \in \BB_1$.
\enprop

\becor Suppose that $R\A$ satisfies (R1), (R2), (R3). Then, for every subspace $\BB\subset R\A$, $\pi^{\BB}_\omega$ is a *-subrepresentation of $\pi^{R\A}_\omega$.
\encor

If $R\A$ satisfies (R1), (R2), (R3),  we write $(\pi_\omega, \lambda_\omega, \H_\omega)$, instead of \linebreak $(\pi^{R\A}_\omega, \lambda^{R\A}_\omega, \H^{R\A}_\omega)$, for the corresponding GNS construction.

\medskip Theorem \ref{GNS} motivates the following definition:
\bedefi Let $\A$ be a partial *-algebra. A linear functional $\omega$  on $\A$ is called {\em representable} if there exists a subspace $\BB$ of $R\A$, with $\BB\neq \{0\}$ and $\BB\neq {\mb C}e$, if $\A$ has a unit $e$, such that the conditions (R1)--(R3) are satisfied.
\findefi

\section{Positive sesquilinear forms with pre-core}\label{sect_biwips}
In \cite{ait_book} positive sesquilinear forms possessing a {\em core} were considered. They were called {\em biweights} and it was shown that a GNS construction can be performed for them. In this section we define the more general notion of pre-core for positive sesquilinear forms and construct the corresponding GNS representation. Before going forth, we review some definitions.

Let $\A$ be a partial *-algebra. Let $\varphi$ be a positive sesquilinear form on
$\D(\varphi) \times \D(\varphi)$, where $\D(\varphi)$ is a nontrivial subspace of $\A$ (i.e., $\D(\varphi)\neq \{0\}$ and $\D(\varphi) \neq {\mb C}e$, if $\A$ has a unit $e$).
 Then we have
\begin{align}
\varphi(x,y) &= \overline{\varphi(y,x)}, \ \ \ \forall \, x, y \in
\D(\varphi)
\\
 |\varphi(x,y)|^2 &\leqslant \varphi(x,x) \varphi(y,y), \ \ \
\forall \, x, y \in \D(\varphi). \label{2.2}
\end{align}
We put
\[
N_\varphi= \{ x \in \D(\varphi) ; \varphi(x,x)=0\}. \] By
\eqref{2.2}, we have
\[
N_\varphi= \{ x \in \D(\varphi) : \varphi(x,y)=0, \ \ \ \forall \,
y \in \D(\varphi) \},
\]
and so $N_\varphi$ is a subspace of $\D(\varphi)$ and the quotient
space $\D(\varphi) / N_\varphi \equiv \{ \lambda_\varphi(x) \equiv
x + N_\varphi ; x \in \D(\varphi) \}$ is a pre-Hilbert space with
respect to the inner product $(\lambda_\varphi(x) |
\lambda_\varphi(y)) = \varphi(x, y), x,y \in \D(\varphi)$. We
denote by $\H_\varphi$ the Hilbert space obtained by the
completion of $\D(\varphi) / N_\varphi$.

\bedefi\label{defn-biw} Let $\varphi$ be a positive sesquilinear
form on $\D(\varphi) \times \D(\varphi)$. A subspace $B(\varphi) $
of $\D(\varphi)$ is said to be a {\it core} for $\varphi$ if

(i) $B(\varphi) \subset R\A$ ;

(ii) $\{ ax ; a \in \A, x \in B(\varphi) \} \subset \D(\varphi)$ ;

(iii) $\lambda_\varphi(B(\varphi))$ is dense in $\H_\varphi$;

(iv) $\varphi(ax, y) = \varphi(x, a\x y), \ \ \forall \, a \in \A,
\forall \, x,y \in B(\varphi)$;

(v) $\varphi(a\x x, by) = \varphi(x, (ab)y), \ \ \forall \, a \in
L(b), \forall \, x,y \in B(\varphi)$. \findefi

\noindent We denote by ${\mc B}_\varphi$ the set of all cores
$B(\vp)$ for $\vp$.

\bedefi
A positive sesquilinear form $\varphi$ on $\D(\varphi)
\times \D(\varphi)$ such that ${\mc B}_\varphi \neq \emptyset$ is
called a {\it biweight} on $\A$.
\\ If $\D(\varphi)=\A$ we call $\vp$ an {\em ips form} (ips stands for invariant positive sesquilinear) on $\A$.
\findefi

To every biweight $\vp$ on $\A$, with core $B(\varphi) $, there corresponds \cite{ait_book}
 a triple $(\pi_\vp^B, \lambda_\varphi, \H_\vp)$, called the {\it GNS
construction} for the biweight $\vp$ on $\A$ with the core $B(\varphi) $, where $\H_\vp$ is a Hilbert space,
$\lambda_\vp$ is a linear map from $B(\varphi) $ into $\H_\vp$, such that $\lambda_\vp(B(\vp))$ is dense in $\H_\vp$, and $\pi_\vp^B$ is a *-representation on $\A$ in the
Hilbert space $\H_\vp$. \\ The representation
$\pi_\vp^B$ is then the closure of the representation $\pi_\vp^\circ$ defined on $\lambda_\vp(B(\vp))$ by
\begin{equation}\label{pizero}
\pi_\vp^\circ(a) \lambda_\vp(x) =\lambda_\vp(ax), \quad a \in \A, x\in B(\vp).
\end{equation}

\medskip

Let hereafter  $\A$ be a semi-associative \pa\ with unit $e$. If $\vp$ is an ips form with core $B(\vp)$, with $B(\vp)$ a subspace of $R\A$ containing the unit $e$ of $\A$, then the linear
functional $\omega_\vp$, with $\omega_\vp(a)=\vp(a,e)$, $a \in \A$, satisfies $B(\vp)\neq {\mb C}e$ and the  conditions (R1),(R2) and
(R3); i.e., it is representable. Thus Theorem \ref{GNS} can be applied to get the *-representation
$\widetilde{\pi}^{B(\vp)}_{\omega_\vp}$ (denoted, hereafter, simply by $\widetilde{\pi}^{B}_{\omega_\vp}$) constructed as shown above. On the other hand, we can also build up, as described above, the closed *-representation $\pi_\vp^B$, with cyclic vector $\xi_\vp:=\lambda_\vp(e)$. Since
$$\omega_\vp(a)=\vp(a,e)=\ip{ \pi_\vp(a)\xi_\vp}{\xi_\vp}, \quad
\forall a\in \A,$$ it turns out that
$\widetilde{\pi}^B_{\omega_\vp}$ and $\pi^B_\vp$ are unitarily
equivalent.

\medskip We now define the notion of pre-core for a positive sesquilinear form on $\A$. This notion is weaker than that of core for a biweight. We then investigate positive sesquilinear forms having a pre-core.

\medskip \bedefi Let $\A$ be a partial *-algebra and $\D(\vp)$ a nontrivial subspace of $\A$. Let  $\vp$ be a positive sesquilinear
form on $\D(\varphi) \times \D(\varphi)$. If  a nontrivial subspace $B(\vp)$ of $\D(\vp)$ satisfies the conditions (i), (ii), (iv) and (v) of Definition \ref{defn-biw}, then it is said to be a {\em pre-core} for $\vp$.\findefi

\betheo\label{lemma_44} Let $\A$ be a semi-associative \pa\ with unit $e$ and $\vp$ a positive sesquilinear form on $\A \times \A$ with a pre-core $B(\vp)$ containing the unit $e$ of $\A$. Then the following statements hold.
\begin{itemize}
\item[(i)]
The linear functional
$\omega_\vp$, with $\omega_\vp(a)=\vp(a,e)$, $a \in \A$, is representable for $B(\vp)$.
\item[(ii)] Put $$\Omega^B_\vp(a,b)=\ip{\widetilde{\pi}^B_{\omega_\vp}(a)\xi_{\omega_\vp}}{\widetilde{\pi}^B_{\omega_\vp}(b)\xi_{\omega_\vp}}, \quad
a,b \in \A,$$
where $\widetilde{\pi}^B_{\omega_\vp}$ is the *-representation of $\A$ defined by $\omega_\vp$ and $\xi_{\omega_\vp}=\lambda^B_{\omega_\vp}(e)$ the corresponding cyclic vector. Then $\Omega^B_\vp$ is an ips form on $\A$ with core $B(\vp)$.

\end{itemize}
\entheo

\begin{proof} (i): This is trivial. \\
(ii): Since $\H^B_{\omega_\vp}= \overline{\lambda^B_{\omega_\vp}(\A)}=\overline{\lambda^B_{\omega_\vp}(B(\vp))}$ and $\widetilde{\pi}^B_{{\omega_\vp}}$ is unitarily equivalent to $\pi^B_{\Omega_\vp}$, it follows that $\H^B_{\Omega_\vp}=\overline{\lambda^B_{\Omega^B_\vp}(B(\vp))}$. Hence, $\Omega^B_\vp$ is an ips form with core $B(\vp)$.

\end{proof}

\begin{lemma} \label{lemma 3.5} The following equality holds
$$\Omega^B_\vp(ax,by) = \vp(ax,by), \quad \forall a,b \in \A \mbox{ such that } a^* \in L(b),\, \forall x,y \in B(\vp).$$
\end{lemma}
\begin{proof} First, we have:
$$ \Omega^B_\vp (ax, y) = \vp(ax,y), \quad \forall a \in \A, \forall x,y \in B(\vp).$$
Indeed,
$$
\Omega^B_\vp (ax, y)
= \omega_\vp (y^*(ax))=\vp(y^*(ax), e)=\vp(ax,y).
$$
By (v) of Definition \ref{defn-biw}, if $a^*\cdot b$ is well-defined, we get
$$ \Omega^B_\vp (ax,b y)=\Omega^B_\vp (x,(a^*b) y)= \vp(x, (a^*b)y)= \vp(ax,by).$$
\end{proof}

\berem
(1): Suppose that $B(\vp)$ is not a core. Then a *-representation like $\pi^B_\vp$ cannot be defined directly through $\vp$, but the *-representation $\pi^B_{\omega_\vp}$ can be defined by means of the representable linear functional $\omega_\vp$.\\
(2): Let $B(\vp)$ be a pre-core for $\vp$, containing $e$.
If $B(\vp)$ is a core for $\vp$, then $\pi^B_\vp$ and $\pi^B_{\Omega_\vp}$ are unitarily equivalent and $\Omega^B_\vp= \vp$.
Suppose, on the contrary, that $B(\vp)$ is not a core for $\vp$. Then, though $\omega_\vp(a)=\vp(a,e)=\Omega^B_\vp(a,e)$, for every $a\in \A$, $\Omega^B_\vp\neq \vp$ and
\begin{equation} \label{5}\H_{\Omega^B_\vp}\cong \overline{\lambda_\vp(B(\vp))} \varsubsetneqq \H_\vp. \end{equation} For an explicit example, see \cite[Example 2.3]{ct_rep}
\enrem
\section{A decomposition theorem}\label{sect_decom}
In this section we define the notion of singularity of a positive sesquilinear form with pre-core and show that every positive sesquilinear form with pre-core can be decomposed into the sum of an ips form and of a singular form.
We first give a necessary and sufficient condition for a pre-core $B(\vp)$ to be a core.

\begin{prop}\label{3.7} Let $\A$ be be a \pa, $\vp$ a positive sesquilinear form on $\A\times \A$ and $B(\vp)$ a pre-core for $\vp$.
The following statements are equivalent.
\begin{itemize}
                                          \item[(i)]  $B(\vp)$ is a core for $\vp$
                                          \item[(ii)] $\lambda_\vp(B(\vp))$ is dense in $\H_\vp$.
                                          \item[(iii)]If $\{a_n\}$ is a sequence of elements of $\A$ such that:
                                          \begin{itemize}
                                            \item[(iii.a)]$\vp(a_n,x)\to 0$, as $n \to\infty$, for every $x \in B(\vp)$;
                                            \item[(iii.b)]$\vp(a_n-a_m,a_n-a_m)\to 0$, as $n,m \to\infty$;
                                          \end{itemize}
                                            then,  $\displaystyle \lim_{n\to \infty}\vp(a_n,a_n) =0$.
                                        \end{itemize}

\end{prop}
\begin{proof} (i) $\Leftrightarrow$ (ii): This follows from the definition.

(ii) $\Rightarrow$ (iii): Let $\{a_n\}$ be a sequence of elements of $\A$ for which (iii.a) and (iii.b) hold. By (iii.b) the sequence $\{\lambda_\vp(a_n)\}$ is Cauchy in $\H_\vp$. Let $\xi$ be its limit. By (iii.a) it follows that
$$ \ip{\xi}{\lambda_\vp(x)}=\lim_{n\to \infty} \vp(a_n,x)= 0, \quad \forall x \in B(\vp).$$
Hence, $\xi$ is orthogonal to $\lambda_\vp(B(\vp))$. This implies that $\xi=0$ and, therefore,
$$ \lim_{n\to \infty} \vp(a_n,a_n)=\|\xi\|^2=0.$$
(iii) $\Rightarrow$ (ii): Let $\xi\in \H_\vp$ be a vector orthogonal to $\lambda_\vp(B(\vp))$ and $\{a_n\}$ a sequence in $\A$ such that $\lambda_\vp(a_n)\to \xi$. Then, it is easily seen that $\{a_n\}$ satisfies (iii.a) and (iii.b). Then,
$$ \|\xi\|^2= \lim_{n\to \infty} \vp(a_n,a_n)=0.$$ This proves that $\lambda_\vp(B(\vp))$ is dense in $\H_\vp$.

\end{proof}

The above statements suggest the following definition of {\em singularity} of a positive sesquilinear form.

\bedefi Let $\A$ be be a \pa\ and
$\psi$ a positive sesquilinear form on $\A\times \A$ with pre-core $B(\psi)$. We say that $\psi$ is $B(\psi)$-{\em singular} if there exists $a_0 \in \A$ with $\psi(a_0,a_0)>0$ and $\psi(a,x)=0$, for every $a\in \A$, $x\in B(\psi)$.

\findefi

\betheo \label{3.9} Let $\A$ be a semi-associative \pa\ with unit $e$ and
$\vp$ a positive sesquilinear form on $\A\times \A$ with pre-core $B(\vp)$. Then, there exist an ips form $\vp_0$, with core $B(\vp)$, which coincides with $\vp$ on all pairs $(a,b)$ with $a^*\in L(b)$, and a  $B(\vp)$-singular form $s_\vp$ such that
$$ \vp(a,b)=\vp_0(a,b)+s_\vp(a,b), \quad \forall a,b \in \A.$$
\entheo
\begin{proof} Put $\vp_0=\Omega^B_\vp$ and $s_\vp= \vp - \Omega^B_\vp$. Then, by Lemma \ref{lemma_44}, $\vp_0$ is an ips form with core $B(\vp)$ and by Lemma \ref{lemma 3.5}, $s_\vp(a,x)=0$, for every $a\in \A$, $x\in B(\vp)$.\\
Now we prove that $\Omega^B_\vp (a,a) \leq \vp(a,a)$, for every $a \in \A$. Indeed we have
$$ \Omega^B_\vp(a,a) = \|\widetilde{\pi}^B_{\omega_\vp}(a)\lambda^B_{\omega_\vp}(e)\|^2= \|\lambda^B_{\omega_\vp}(a)\|^2.$$
Now, we notice that, by the construction in Theorem \ref{GNS}, (see, in particular \eqref{eqn_defn_lambda}, \eqref{eqn_defn_lambda2}\,),
\begin{eqnarray*}\|\lambda^B_{\omega_\vp}(a)\|&=& \sup\{|\omega_\vp(a^*x)|; \omega_\vp(x^*x)=1\}\\ &=& \sup\{|\vp(a,x)||; \omega_\vp(x^*x)=1\}\leq \vp(a,a)^{1/2}.\end{eqnarray*}
Hence, $s_\vp$ is a positive sesquilinear form on $\A\times \A$ and $B(\vp)$ is a pre-core also for $s_\vp$. If $B(\vp)$ is not a core for $\vp$, then there exists a sequence $\{a_n\}$ of elements of $\A$ with the properties
(a): $\vp(a_n,x)\to 0$, as $n \to\infty$, for every $x \in B(\vp)$; (b): $\vp(a_n-a_m,a_n-a_m)\to 0$, as $n,m \to\infty$; (c): $\lim_{n\to \infty}\vp(a_n,a_n)= \alpha>0$. Since $B(\vp)$ is a core for $\Omega^B_\vp$ and
$$\Omega^B_\vp((a_n-a_m,a_n-a_m)\leq \vp(a_n-a_m,a_n-a_m) \to 0,$$ we have $\lim_{n\to \infty}\Omega^B_\vp(a_n,a_n)=0$.
In conclusion,  we have $ s_\vp(a_n,a_n)\to \alpha>0$. So that $s_\vp$ cannot be identically $0$. This, clearly, implies that $s_\vp$ is $B(\vp)$-singular.
\end{proof}

\berem  Proposition 2.2 of \cite{ct_rep} was stated in incorrect way. The right version can be recovered specializing to the case of quasi *-algebras  the above Proposition \ref{3.7}.
\enrem

\section{Quasi-regular *-representations}\label{sect_quasireg}
Let $\A$ be a \pa, with unit $e$, such that $R\A \varsupsetneqq {\mb C}e$. Let  $\pi$ be a *-representation of $\A$
into $ \L\ad(\D,\H)$ and $\xi \in \D$. Put
$$\vp_\xi^\pi(a,b) = \ip{\pi(a) \xi} {\pi(b)\xi}, \quad a,b \in \A.$$
Then $\vp_\xi^\pi$ is a positive sesquilinear form on $\A\times \A$  and the subspace
$$ B_0(\vp_\xi^\pi)=\{ x \in R\A:\, \pi(x)\xi \in {\D}(\pi)\} $$ is a pre-core  for $\vp_\xi^\pi$, containing $e$ and it is the largest member in the set of all pre-cores for $\vp_\xi^\pi$.
Thus, if $\pi(B_0(\vp_\xi^\pi))\xi$ is dense in $\overline{\pi(\A)\xi}$, then $\vp_\xi^\pi$ is an ips form on $\A$.

More in general, a vector form related to $\pi$ can be defined for every $\xi\in \H$.
This is done by defining
\begin{align*}
\D(\vp_\xi^\pi)&=\{a \in \A: \xi \in \D(\pi(a\ha)\ha)\} \\
\vp_\xi^\pi(a,b)& = \ip{\pi(a\ha)\ha \xi} {\pi(b\ha)\ha\xi}.
\end{align*}
Then $\vp_\xi^\pi$ is a positive sesquilinear form on
$\D(\vp_\xi^\pi)\times \D(\vp_\xi^\pi)$.

If $\xi \in \H \setminus \D$ we define $B_0(\vp_\xi^\pi)$ as the
linear span of the set
$$ B_{00}(\vp_\xi^\pi)= \{ x \in R\A:\, \xi \in \D(\overline{\pi(x)}),
\,\overline{\pi(x)}\xi \in \D\}.$$ Then, $B_0(\vp_\xi^\pi)$ is a pre-core  for $\vp_\xi^\pi$, that does not contain $e$. If, in addition,
$\pi(B_0(\vp_\xi^\pi))\xi$ is dense in
$\overline{\pi(D(\vp_\xi^\pi))\xi}$, $\vp_\xi^\pi$ is a biweight
on $\A$ (see \cite[Example
9.1.12]{ait_book}).

 The previous discussion shows that the set of pre-cores for $\vp_\xi^\pi$ is not empty, for every $\xi\in \H$. The existence of cores for $\vp_\xi^\pi$ remains an open question.

 The notion of {\em regular *-representation} has been given in \cite{anttratschi_2, ct_rep}. We weaken it a little to get a notion more suitable for our purposes.
 \bedefi A *-representation $\pi$ of a partial
*-algebra $\A$  is said to be
\begin{itemize}
  \item {\em regular}, if
$\vp_\xi^\pi$ is a biweight on $\A$, for every $\xi \in \H$;
  \item {\em quasi regular}, if
$\vp_\xi^\pi$ is an ips form on $\A$, for every $\xi \in \D$.
\end{itemize}
\findefi

\beprop \label{CNSregularity}Let $\pi$ be a *-representation of
$\A$. The following statements are equivalent.
\begin{itemize}
\item[(i)] $\pi$ is quasi regular.
\item[(ii)] There exists a pre-core $B(\vp)$ such that $\pi(a)\xi \in \overline{\pi(B(\vp))\xi}$, for every $a \in
\A$ and for every $\xi \in \Dr$.
\item[(iii)] There exists a pre-core $B(\vp)$ such that for every $\xi \in \Dr$,
$\pi\ze:=\pi_{\upharpoonright{\mc M}_\xi}$ is a *-representation
of $\A$ into ${\mc L}^\dag ({\mc M}_\xi, \overline{{\mc
M}_\xi})$, where ${\mc M}_\xi=\pi(B(\vp))\xi$.
\end{itemize}\enprop

\begin{proof} (i)$\Rightarrow$(ii): Let $\pi$ be quasi regular and
$\xi\in \Dr$. Let us consider the vector form $\vp_\xi^\pi$. Then, for
every $a \in \A$, there exists a sequence $\{x_n\} \subset B_0(\vp_\xi^\pi)$
such that $\|\lambda_{\vp_\xi^\pi}(a-x_n)\|\to 0$. Then we have:
$$ \|(\pi(a)-\pi(x_n))\xi\|^2 = \|\lambda_{\vp_\xi^\pi}(a-x_n)\|^2 \to
0.$$ This proves that $\pi(a)\xi \in \overline{\pi(B_0(\vp_\xi^\pi))\xi}$.

\noindent (ii)$\Rightarrow$(iii): The assumption implies that, for
every $a \in \A$ and $\xi \in \Dr$, $\pi\ze(a)$ maps $\pi(B(\vp))\xi$
into $\overline{\pi(B(\vp))\xi}$. Some simple calculations, that make
use of conditions (iv) and (v) of Definition \ref{defn-biw},  show that
$\pi\ze(a^*)=(\pi\ze(a))^*_{\upharpoonright\pi(B(\vp))\xi}$ and that
$\pi\ze$ preserves the partial multiplication of $\A$.

\noindent (iii)$\Rightarrow$(i): The assumption implies that, for
every $\xi \in \Dr$ and $a \in \A$, $\pi(a)\xi \in \overline{{\mc
M}_\xi}$. Therefore, for every $a \in \A$, there exists a sequence
$\{x_n\} \subset B(\vp)$ such that
$\mbox{$\|(\pi(a)-\pi(x_n))\xi\|\to 0$}$. Then, for $\vp_\xi^\pi$, we
have:
$$\vp_\xi^\pi(a-x_n,a-x_n) =\|\lambda_{\vp_\xi^\pi}(a-x_n)\|^2=\|(\pi(a)-\pi(x_n))\xi\|^2 \to
0.$$ Hence, $\pi$ is quasi regular.
\end{proof}

\medskip
If $\vp$ is a positive sesquilinear form on $\A \times \A$ and $B(\vp) \subset R\A$, then, for every $x \in B(\vp)$,  the sesquilinear form $\vp_x$ on $\A \times \A$
defined by \begin{equation}\label{eq_defvpa}\vp_x (a,b) = \vp(ax,bx), \quad a,b \in
\A,\end{equation}is a positive sesquilinear form on $\A \times \A$.
Now assume that $\A$ is semi-associative and that  $B(\vp)$ is a pre-core for $\vp$. If $B(\vp)$ is an algebra, then, for every $x \in B(\vp)$, $\vp_x$ admits $B(\vp)$ as a pre-core.
\medskip
\beprop \label{regularity} Let $\A$ be a  semi-associative partial *-algebra, $\vp$ an ips form with core $B(\vp)$ and $\pi_\vp^\circ $ the
*-representation defined in \eqref{pizero}. Suppose that $B(\vp)$ is an algebra. Then, the following statements
are equivalent:
\begin{itemize}
\item[(i)] $\pi_\vp^\circ $ is quasi regular.
\item[(ii)] $B(\vp)$ is a core for $\vp_x$, for every $x \in B(\vp)$.
\end{itemize} \enprop
\begin{proof}  If $\eta \in \D_\vp=\lambda_\vp(B(\vp))$, then $\eta= \lambda_\vp(x)$,
for some $x \in B(\vp)$.  Hence $$ \vp_\eta^{\pi_\vp^\circ}(a,b)=
\ip{\pi_\vp^\circ(a)\lambda_\vp(x)}{\pi_\vp^\circ(b)\lambda_\vp(x)}=\vp(ax,ax)=
\vp_x(a,b), \quad \forall a,b \in \A.
$$
Thus $\vp_\eta^{\pi_\vp^\circ}=\vp_x$. This equality clearly implies the
equivalence of (i) and (ii).
\end{proof}


\begin{thebibliography} {99}

\bibitem{ait_book}  {\sc    J-P. Antoine, A. Inoue, and C. Trapani},
{\it Partial *-Algebras and Their Operator Realizations}, Kluwer,
Dordrecht, 2002.

\bibitem{at_pipbook} {\sc    J-P. Antoine and C. Trapani}, {\it Partial Inner Product Spaces: Theory and Applications}, Springer Lecture notes in Mathematics, to appear.

\bibitem{anttratschi_2}{\sc J-P. Antoine, C. Trapani,  and F. Tschinke}, \textit{Continuous *-homomorphisms of Banach partial
*-algebras}, Mediterranean J.Math. {\bf 4} (2007), 357-373
\bibitem{ct_rep} {\sc C. Trapani}, {\it *-Representations, seminorms and structure properties of normed quasi *-algebras}, Studia Mathematica {\bf 186} (2008) 47-75


\end{thebibliography}
\end{document}